\documentclass[11pt]{article}
\usepackage{fullpage}
\usepackage{amsmath}
\usepackage{amssymb}
\usepackage{graphicx}
\usepackage{psfrag}

\newtheorem{definition}{Definition}
\newtheorem{theorem}{Theorem}

\newtheorem{lemma}{Lemma}
\newenvironment{proof}{\noindent \begin{rm}{\textbf{Proof.} }}{\hspace*{\fill}$\Box$\par\end{rm} \vspace{.3cm}}

\newcommand{\id}{\mbox{ID}}

\newcommand{\InfoMsg}{\mbox{\tt InfoMsg}}

\newcommand{\parent}{\mbox{\sf parent}}
\newcommand{\level}{\mbox{\sf level}}
\newcommand{\dist}{\mbox{\sf dist}}
\newcommand{\need}{\mbox{\sf need}}
\newcommand{\connect}{\mbox{\sf connected}}
\newcommand{\member}{\mbox{\sf member}}
\newcommand{\connectpt}{\mbox{\sf connect\_pt}}

\newcommand{\IsRoot}{\mbox{\sf Is\_Root}}
\newcommand{\CRoot}{\mbox{\sf CRoot}}
\newcommand{\CParent}{\mbox{\sf CParent}}
\newcommand{\distConnect}{\mbox{\sf distConnect}}
\newcommand{\distNotConnect}{\mbox{\sf distNotConnect}}
\newcommand{\parentConnect}{\mbox{\sf parentConnect}}
\newcommand{\parentNotConnect}{\mbox{\sf parentNotConnect}}
\newcommand{\Better}{\mbox{\sf Better\_Path}}
\newcommand{\Ask}{\mbox{\sf Asked\_Connection}}
\newcommand{\ConnectS}{\mbox{\sf Connect\_Stab}}
\newcommand{\ConnectPtS}{\mbox{\sf Connect\_Pt\_Stab}}

\newcommand{\ARA}{$\mathcal{DR}_1$}
\newcommand{\ARB}{$\mathcal{DR}_2$}
\newcommand{\CRA}{$\mathcal{RR}$}
\newcommand{\CRB}{$\mathcal{NR}_1$}
\newcommand{\CRC}{$\mathcal{NR}_2$}
\newcommand{\CRD}{$\mathcal{CR}_1$}
\newcommand{\CRE}{$\mathcal{CR}_2$}
\newcommand{\CRF}{$\mathcal{TR}$}
\newcommand{\CRG}{$\mathcal{CR}_3$}

\newcommand{\STT}{{\sc s3t}}

\newenvironment{lemma-repeat}[1]{\begin{trivlist}
\item[\hspace{\labelsep}{\bf\noindent Lemma~\ref{#1} }]}%
{\end{trivlist}}
\newcommand{\toto}{xxx}

\begin{document}
\title{A Superstabilizing $\log(n)$-Approximation Algorithm for Dynamic Steiner Trees}

\author{
L\'{e}lia Blin$^{1,2}$ 
\and
Maria Gradinariu Potop-Butucaru$^{2,3}$
\and
St\'{e}phane Rovedakis$^1$
}

\footnotetext[1]{Universit\'e d'Evry, IBISC, CNRS, France.}
\footnotetext[2]{Univ. Pierre \& Marie Curie - Paris 6, LIP6-CNRS UMR 7606, France.}
\footnotetext[3]{INRIA REGAL, France.}

\date{}

\maketitle

\begin{abstract}
In this paper we design and prove correct a fully dynamic distributed algorithm for maintaining an approximate Steiner 
 tree that connects via a minimum-weight spanning tree a 
subset of nodes of a network (referred as Steiner members or Steiner group) .
Steiner trees are good 
candidates to efficiently implement communication primitives such
as publish/subscribe or multicast, essential building blocks for the new
emergent networks (e.g. P2P, sensor or adhoc networks). 

The cost
of the solution returned by our algorithm is at most $\log |S|$
times the cost of an optimal solution, where $S$ is the group of
members. Our algorithm improves over existing solutions in several
ways. First, it tolerates the dynamism of both the group members and
the network. Next, our algorithm is self-stabilizing, that is, it
copes with nodes memory
corruption. Last but not least, our algorithm is
\emph{superstabilizing}. 
That is, while converging to a correct
configuration (i.e., a Steiner tree) after a modification of the
network, it keeps offering the Steiner tree service during the
stabilization time to all members that have not been affected by
this modification.
\end{abstract}


\section{Introduction}

The design of efficient distributed applications in the newly distributed 
emergent networks such as MANETs, P2P or sensor networks raises 
various challenges ranging 
from models to fundamental services. These networks face frequent
churn (nodes and links creation or destruction) and various privacy
and security attacks that cannot be easily encapsulated in the
existing distributed models. 
Therefore, new models and new algorithms have to be designed. 

Communication services are the building blocks for any distributed
system and they have received a particular attention in the lately
years. Their efficiency greatly depends 
on the performances of the underlying routing overlay. 
These overlays should be optimized to reduce the network
overload. Moreover, in order 
to avoid security and privacy attacks the number of network nodes that
are used only for the overlay connectivity have to be minimized. Additionally, the overlays have to offer 
some quality of services while nodes or links fail.
  
The work in designing optimized communication overlays for the new emergent
networks has been conducted in both structured (DHT-based) and un-structured networks. 
Communication primitives using DHT-based schemes such as Pastry, CAN or Chord  \cite{CDCHR03} 
build upon a global naming scheme based on hashing nodes identifiers.
These schemes are optimized to efficiently route in the virtual name
space however they have weak energy performances in MANETs or
sensor networks where the maintenance of long links reduces the
network perennial. Therefore, alternative strategies \cite{KS07}, mostly based on
gossip techniques, 
have been recently considered. These schemes, highly 
efficient when nodes have no information on the content and the topology of the system,
offer only probabilistic guarantees on the message delivery.

In this paper we are interested in the study of overlays targeted 
to efficiently connect a group of nodes that are 
not necessarily located in the same geographical area (e.g. sensors that should communicate their 
sensed data to servers located outside the deployment area, P2P nodes that share the same 
interest and are located in different countries, robots that should 
participate to the same task but need to remotely coordinate). Steiner trees are good 
candidates to implement the above mentioned requirements since the problem have been 
designed for efficiently connect a subset of the network nodes, referred as Steiner members.

\paragraph{The Steiner tree problem.}

The Steiner tree problem can be informally expressed as follows: 
given a weighted graph 
in which a subset $S$ 
of nodes 
is identified, 
find a minimum-weight tree spanning $S$.
The Steiner tree problem is one of the most important 
combinatorial optimization problems and finding a Steiner tree is NP-hard. 



A survey on different heuristics for constructing Steiner trees with different competitiveness levels can be found
in \cite{thesis96,Steiner_survey}.
In our work we are interested in dynamic variants of Steiner trees
first addressed in \cite{ImaseWaxman91} in a centralized online setting.
They propose a $\log |S|$-approximation algorithm for this problem that copes
only with Steiner member arrivals. This algorithm can be 
implemented in a decentralized environment (see~\cite{GataniRG05}).

Our work considers the fully dynamic version
of the problem where both Steiner members and ordinary nodes can join
or leave the system.
Additionally, our work aims at providing a 
superstabilizing approximation of a Steiner tree. 
The property of self-stabilization  \cite{Dij74,Dol00} enables a distributed
algorithm to recover from a transient fault regardless of its initial
state. The superstabilization \cite{Dolev_SuperStab} is an extension of the
self-stabilization property for dynamic settings. The idea is to
provide some minimal guarantees while the system repairs after a 
topology change.

To our knowledge there are only two self-stabilizing approximations of
Steiner trees
\cite{SS_Steiner02,Kakugawa_Steiner_journal}. Both works assume the 
shared memory model and an unfair centralized scheduler.
In \cite{SS_Steiner02} the authors propose a self-stabilizing
algorithm based on a pruned minimum spanning 
tree. 
The computed solution has a competitiveness of
$|V|-|S|+1$ where $V$ is the set of nodes in the network. 
In \cite{Kakugawa_Steiner_journal}, the authors proposed a
four-layered algorithm that builds upon the techniques 
proposed in \cite{WuWW86} in order to obtain a $2$ approximation. 


The above cited algorithms work only for static networks. 

\paragraph{Our results.}

We describe a super-stabilizing algorithm for the Steiner tree problem. This algorithm has the following properties:
\begin{itemize}
\item First, it is distributed, i.e.,  completely decentralized. That is, nodes
locally self-organize in a Steiner tree. The cost of the constructed Steiner tree 
is at most $\log |S|$ times the cost of an optimal solution, where $S$ is the Steiner group. 
\item Second, our algorithm is specially designed to cope with user
dynamism. In other words, 
our solution withstand when nodes (or links) join and leave the system. 
\item Third, our algorithm includes self-stabilization policies. Starting from an arbitrary state (nodes local memory corruption, counter program corruption, or
erroneous messages in the network buffers), our algorithm is guaranteed to converge to a tree spanning the Steiner members. 
\item Fourth, our algorithm is \emph{superstabilizing}. That is, while
a topology change occurs, i.e., during the restabilization period,
 the algorithm offers the guarantee that only the subtree connected
through the crashed node/edge is reconstructed. 
\end{itemize}

\begin{table*}[!htb]
\begin{center}
\scalebox{1}
{
\begin{tabular}{|l|c|c|c|}
\hline
 & Approximation ratio & Self-Stabilizing & Superstabilizing \\
\hline
Chen et al. \cite{ChenHK93} & 2 & No & No \\
Kamei and Kakugawa \cite{SS_Steiner02} & $|V|+|S|-1$ & Yes & No \\
Kamei and Kakugawa \cite{Kakugawa_Steiner_journal} & $2$ & Yes & No \\
\hline
This paper &  $O(\log |S|)$ & Yes & Yes \\
\hline
\end{tabular}
}
\caption{Distributed (deterministic) algorithms for the Steiner tree problem.}
\label{tableresume}
\end{center}
\end{table*}

Table~\ref{tableresume} summarizes our contribution compared to
previous works. Hence, our algorithm is the first superstabilizing
algorithm for the Steiner tree problem. Its approximation ratio is
logarithmic, which is not as good as the 2-approximation algorithm by
Kamei and Kakugawa in \cite{Kakugawa_Steiner_journal}. However, this
latter algorithm is not superstabilizing. Designing a superstabilizing
2-approximation algorithm for the Steiner tree problem is a
challenge. Indeed, all known 2-approximation distributed algorithms
(self-stabilizing or not) for the Steiner tree problem use a minimum
spanning tree (MST), and the design of a superstabilizing algorithm
for MST is a challenge by itself.


\section{Model and notations} 
\label{sec:model}
We consider an undirected weighted connected network $G=(V,E,w)$ where
$V$ is the set of nodes, $E$ is the set of edges and $w: E \rightarrow
{\mathbb R}$ is a cost function. Nodes represent processors and edges
represent bidirectional communication links. Each node in the network
has an unique identifier. 
$S \subseteq V$ defines the set of members we have to connect. For any pair of nodes $u,v \in V$, we note $d(u,v)$ the distance of the shortest path $P(u,v)$ between $u$ and $v$ in $G$ (i.e. $d(u,v)=\sum_{e \in P(u,v)} w(e)$). For a node $v \in V$, we denote the set of its neighbors ${\mathcal N(v)}=\{u, (u,v) \in E\}$. A Steiner tree, $T$ in $G$ is a connected acyclic sub-graph of $G$ such that $T=(V_{T},E_{T})$, $S \subseteq V_{T} \subseteq V$ and $E_{T} \subset E$. We denote by $W(T)$ the cost of a tree $T$, i.e. $W(T)=\sum_{e \in T} w(e)$.


We consider an asynchronous communication message passing model with
FIFO channels (on each link messages are delivered in the same order as
they have been sent).

A {\it local state} of a node is the value of the local variables of
the node and the state of its program counter. 
We consider a fined-grained communication atomicity model \cite{BK07,Dol00}. That is,
each node maintains a local copy of the variables of its
neighbors. These variables are refreshed via special
messages (denoted in the sequel $\InfoMsg$) exchanged periodically 
by neighboring nodes.
A {\it configuration} of the
system is the cross product of the local states of all nodes in the
system plus the content of the communication links. 
The transition from a configuration to the next one is produced by the
execution of an atomic step at a node. An {\it atomic step} at node $p$ 
is an internal computation based on the
current value of $p$'s local variables and a single communication
operation (send/receive) at $p$. 
An {\it execution} of the system is an infinite sequence of
configurations, $e=(c_0, c_1, \ldots c_i, \ldots)$, where each
configuration $c_{i+1}$ follows from $c_i$ by the execution 
of a single atomic step.

In the sequel we consider the system can start in any
configuration. That is, the local state of a node can be corrupted.
Note that we don't make any assumption on the bound of
corrupted nodes. In the worst case all nodes in the system 
may start in a corrupted configuration.
In order to tackle these faults we use {\it self-stabilization} techniques.

Given $\mathcal{L_{A}}$ a non-empty
\emph{legitimacy predicate}\footnote{A legitimacy predicate is 
defined over the configurations of a system and 
is an indicator of its correct behavior.} 
an algorithm $\mathcal{A}$ is \emph{self-stabilizing} 
iff
the following two conditions hold:
\textsf{(i)} Every computation of $\mathcal{A}$ starting from a configuration
satisfying $\mathcal{L_A}$ preserves $\mathcal{L_A}$ (\emph{closure}).  
\textsf{(ii)} Every computation of $\mathcal{A}$ starting from an arbitrary configuration
contains a configuration that satisfies $\mathcal{L_A}$
(\emph{convergence}).

A {\it legitimate configuration} for the Steiner Tree is a
configuration that provides an instance of a tree $T$ spanning $S$. 
Additionally, we expect a competitiveness of $\log(z)$, 
i.e. $\frac{W(T)}{W(T^*)} \leq \log(z)$, with $|S|=z$ and $T^*$ an optimal Steiner tree.

In the following we propose a self-stabilizing Steiner tree algorithm.
We expect our algorithm to be also {\it superstabilizing}
\cite{Dolev_SuperStab}. That is, 
given a class of topology changes $\Lambda$ and a passage predicate, an algorithm 
is superstabilizing with respect to $\Lambda$ iff it
is self-stabilizing, and for every
execution\footnote{\cite{Dolev_SuperStab} use the notion of trajectory
which is the execution of a system enriched with dynamic actions.} $e$ beginning at a
legitimate state and containing a single topology change event of type
$\Lambda$, the passage predicate holds for every configuration in $e$.





In the following we propose a self-stabilizing Steiner tree algorithm and
extend it 
to a superstabilizing Steiner tree algorithm
that copes with the Steiner members and tree edges removal. During the tree
restabilization the algorithm verifies a passage predicate detailed below. 
Then, we discuss the extension of the algorithm to fully dynamic
settings (the add/removal of members, nodes or links join/leave). 
This second extension offers no guarantees during the restabilization period.

\section{The Superstabilizing Algorithm \STT}

The section describes a superstabilizing algorithm for the Steiner tree problem, called \STT.
It implements the
technique proposed by Imase and
Waxman~\cite{ImaseWaxman91},  in a stabilizing manner. That is, each Steiner member is connected
to the existing Steiner tree via a shortest path. Note that in a
stabilizing setting the initial configuration may be arbitrary hence
nodes have to perpetually verify the coherency of their state:
a Steiner member has to be connected to the Steiner
tree via a shortest path while a not Steiner node which does not
serve for the tree connectivity has to be recognized as disconnected.
In our implementation we assume a special node that acts as the root
of the Steiner tree. To this end, we assume an underlying overlay that elects a
leader within the Steiner group. That is, we assume a leader oracle that returns to every node
in the system its status: leader or follower. The leader of the system
is a node in the Steiner group. Note that the implementation of a
leader oracle is beyond the scope of the current work. Several
implementations fault-tolerant and self-stabilizing can be found in
\cite{DDF07}. Recently, algorithms that implement leader
oracles in dynamic settings are proposed in \cite{Baldoni08} for example.

\subsection{Detailed description}


\subsubsection{Variables and Predicates}

For any node $v \in V(G)$, $N(v)$ is the neighbors set of $v$ in the network $G$ (our algorithm is built upon an underlying self-stabilizing protocol that regularly updates the neighbor set of every node). We denote by $\id_v \in \mathbb{N}$ the unique network identifier of $v$. Every node $v$ maintains seven variables for constructing and maintaining a Steiner tree. Three of them are integers, and the others are booleans.

\begin{itemize}
\item $\parent_v$: $\id$ of the parent of node $v$ in the current tree;
\item $\level_v$: number of nodes on the path between the root and $v$
in spanning tree;
\item $\dist_v$: the shortest distance to a node already connected to the current tree;
\item $\member_v$: \emph{true} if $v \in S \subseteq V$, \emph{false} otherwise (this is not a variable wrote by the algorithm but only read);
\item $\need_v$: \emph{true} if $v \in S \subseteq V$ or $v$ has a descendant which is a member, \emph{false} otherwise;
\item $\connect_v$: \emph{true} if $v$ is in the current tree, \emph{false} otherwise;
\item $\connectpt_v$: \emph{true} if $v$ is a member or $v$ has more than one children in the current tree, \emph{false} otherwise.
\end{itemize}

\begin{small}
\begin{figure*}
\scalebox{1}
{
\fbox{
\begin{minipage}{15.7cm}
\footnotesize{
\begin{description}
\item[$\CRoot(v)$] $\equiv \dist_v=0 \wedge \parent_v=\id_v \wedge \need_v \wedge \connect_v \wedge \connectpt_v \wedge \level_v=0$
\item[$\CParent(v)$]$ \equiv (\exists u \in N(v), \parent_v=\id_u)\wedge (\level_v = \level_{\parent_v}+1) \wedge (\not \exists u \in N(v), \parent_u=\id_v \wedge \level_u\neq \level_v+1)$
\item[$\Ask(v)$] $\equiv (\exists u \in N(v), \parent_u=\id_v \wedge \need_u)$
\item[$\Better(v)$]$ \equiv (\neg \connect_v \wedge \dist_v \neq \distNotConnect(v)) \vee (\connect_v \wedge \dist_v \neq \distConnect(v))$
\item[$\ConnectPtS(v)$]$ \equiv (\member_v \wedge \connectpt_v) \vee (\neg \member_v \wedge |\{u:u\in N(v) \wedge \parent_u = \id_v \wedge \connect_u\}|>1)$
\item[$\ConnectS(v)$]$ \equiv \need_v \wedge \connect_{\parent_v} \wedge [\member_v \vee (\neg \member_v \wedge \Ask(v))]$
\item[$\distNotConnect(v)$]$\equiv \min(\min\{w(u,v): u \in N(v) \wedge \connect_u\} ,\dist_u+w(u,v): u \in N(v) \wedge \neg \connect_u\})$
\item[$\parentNotConnect(v)$]$\equiv \arg (\distNotConnect(v))$
\item[$\distConnect(v)$]$ \equiv \min(\min\{w(u,v): u \in N(v) \wedge \connect_u \wedge \connectpt_u\}, \min\{\dist_u+w(u,v): u \in N(v) \wedge [\neg \connect_u \vee (\connect_u \wedge \neg \connectpt_u))\}$
\item[$\parentConnect(v)$]$\equiv \arg (\distConnect(v) )$
\end{description}
}
\end{minipage}
}
}
\caption{Predicates used by the algorithm.}
\label{fig:algo_predicates}
\end{figure*}
\end{small}

\subsubsection{Description of the algorithm}
\label{sec:algo}
Every node $v \in V$ sends periodically its local variables to each of its neighbors
using $\InfoMsg$ messages. Upon the reception of this message a neighbor updates the local copy
of its neighbor variables. The description of a $\InfoMsg$ message is
as follows:\\ 
$\InfoMsg_v[u]=\langle\InfoMsg,\parent_v,\level_v,\dist_v,\need_v,\connect_v,\connectpt_v \rangle.$

Our algorithm is a four phase computation: (1)
first nodes update their distance to the existing Steiner tree, then (2) nodes
request connection (if they are members or they received a connection
demand), then (3) they establish the
connection, and finally (4) they update the state of the current Steiner tree.  These
phases have to be performed in the given order.
That is, a node cannot initiate a request for connection for example if it has not yet updated
its distance.

Note that if a node detects a
distance modification in its neighborhood, it can change its
connection to the current tree. Therefore a node  before computing any
other action must update its
distance to the current tree. 

Every node in the network, 
maintains a
parent link. The parent of a node is one of its neighbors having the shortest
distance to the current tree. Note that erroneous initial configurations may create cycles in the
parent link. To break these cycles, we 
use the notion of tree level, defined by the variable \level: 
the root has the level zero and each node has the level equal to its parent level plus one.

When a member tries to connect to the
tree, it sets its variable
\need\/ to \emph{true}. 
When a node in the current tree receives a
demand for connection, an acknowledgment is sent back along the
requesting path enabling every node along this path to set a variable \connect\/ to
\emph{true}. Nodes with \connect\/ set \emph{true} are called
``connected nodes''. 

Whenever a node detects an incoherency in its neighborhood it disconnects
from the current tree.

In order to give a $\log(z)$-approximate Steiner tree, we introduce a variable \connectpt. This variable
signals if a node is a connection point or not. A connection
point is a connected node which is a member or has more than one
connected child. 


\paragraph{\underline{Algorithm:}}
Upon the reception  of a $\InfoMsg$ nodes correct their local state
via the rules explained below then broadcast their new local state in 
their local neighborhood.
%



\paragraph{\underline{Root:}}

In a coherent state the root has a distance and a level equal to zero,
variables $\need$ and $\connect$ 
are \emph{true} since the root is always connected (it always belongs to the Steiner tree). Variable $\connectpt$ is \emph{true} because the root is a member so a connection point.
Whenever the state of the root is incoherent the Rule \CRA\/ below is enabled.
\begin{small}
\begin{description}
\item[\CRA: (Root reinitialization)]~\\
\textbf{If} $\IsRoot(v) \wedge \neg \CRoot(v)$ \textbf{then}\\
\hspace*{0,8cm}$\dist_v:=0;$ $\parent_v:=\id_v;$ $\need_v:=true\/; $ $\connect_v:=true;$\\ 
\hspace*{0,8cm}$\connectpt_v:=true;$ $ \level_v:=0;$
\end{description}
\end{small}

\paragraph{\underline{Distance update:}} 

Rule \ARA\/ enables to a not connected node to compute its shortest path
distance to the Steiner tree 
as follows: Take the minimum between the edge weights with connected
neighbors and the distances 
with not connected neighbors. If a not connected node detects it has a
better 
shortest path (see Predicate $\Better$) then it updates its distance 
(using Predicates \distNotConnect\/ and \distConnect) and changes its
other variables accordingly. 

The same rule is used to reinitiate the state of a node if it observes that its parent is no more in its neighborhood.

Similarly, Rule \ARB\/ enables to a connected node to compute its
shortest path distance. In order to execute this rule 
a connected node must have a stabilized connection. The distance is
computed as for a not connected node but a connected node compares
this distance with its local distance towards its connection point and takes the minimum.


\begin{small}
\begin{description}
\item[\ARA : (Distance stabilization for not connected nodes)]~\\
\textbf{If} $\neg \IsRoot(v) \wedge [(\neg \connect_v \wedge \Better(v)) \vee \neg \CParent(v)]$ \textbf{then}\\ 
\hspace*{0,8cm}$\dist_v:=\distNotConnect(v);$ $\parent_v:=\parentNotConnect(v);$\\
\hspace*{0,8cm}$ \connect_v:=false;\connectpt_v:=false;$ $\level_v:=Level_{\parent_v}+1;$

\item[\ARB: (Distance stabilization for connected nodes)]~\\
\textbf{If} $\neg \IsRoot(v) \wedge \connect_v \wedge \ConnectS(v) \wedge \Better(v) \wedge \CParent(v) \wedge \ConnectPtS(v)$ \textbf{then}\\ \hspace*{0,8cm}$\dist_v:=\distConnect(v);$ $\parent_v:=\parentConnect(v);$\\
\hspace*{0,8cm}$ \level_v:=Level_{\parent_v}+1;$
\end{description}
\end{small}
\paragraph{\underline{Request to join the tree:}}

Variable $\need$ is used by a not connected node to ask to its parent
a connection to the current Steiner tree. 
Since a member must be connected to the Steiner tree, each member sets
this variable to \emph{true} using Rule \CRB. 
A not member and not connected node which detects that a child wants
to be connected (see Predicate $\Ask$) 
changes its variable $\need$ to \emph{true}. This connection request
is forwarded in the spanning 
tree until a not connected node neighbor of a connected node is
reached.

A not connected node sets its variable $\need$ to \emph{false} using
Rule \CRC\/ if it is not a member and it has no child requesting a
connection. 


%
\begin{small}
\begin{description}
\item[\CRB: (Nodes which need to be connected)]~\\
\textbf{If} $\neg \IsRoot(v) \wedge \neg \need_v \wedge \neg \connect_v \wedge \neg \Better(v) \wedge \CParent(v) \wedge [\member_v \vee (\neg \member_v \wedge \Ask(v))]$\\
\textbf{then} $\need_v:=true;$
\item[\CRC: (Nodes which need not to be connected)]~\\
\textbf{If} $\neg \IsRoot(v) \wedge \neg \connect_v \wedge \need_v \wedge \neg \member_v \wedge \neg \Ask(v) \wedge \neg \Better(v) \wedge \CParent(v)$ \\
\textbf{then} $\need_v:=false;$
\end{description}
\end{small}
\paragraph{\underline{Member connection:}}

When a not connected node neighbor of a connected node (i.e. which
belongs to the Steiner tree) detects a connection request from a child
(i.e. Predicate $\Ask$ is \emph{true}), an acknowledgment is sent
backward using variable $\connect$ along the request path. Therefore
every not connected node 
on this path uses Rule \CRD\/ and sets $\connect$ to \emph{true} until
the member that asked the connection is connected. Only a node that
has (1) no better path, 
(2) its variable $\need=true$ and (3) a connected parent can use Rule
\CRD.

A connected node becomes not connected if its connection path is no
more stabilized (i.e. Predicate $\ConnectS$ is false). Therefore, it
sets $\connect$ to \emph{false} using Rule \CRE.

The parent distance is used for the disconnection of a subtree whenever a
fault occurs in the network. 
If a fault occurs (parent distance is infinity), a connected node in the subtree below a faulty node or edge in the spanning tree must be disconnected using Rule \CRG. So the node sets $\connect$ to false and $\dist$ to infinity and waits until all its subtree is disconnected (i.e. it has no connected child).

%

\begin{small}
\begin{description}
\item[\CRD: (Nodes which must be connected)]~\\
\textbf{If} $\neg \IsRoot(v) \wedge \neg \connect_v \wedge \ConnectS(v) \wedge \neg \Better(v) \wedge \CParent(v)$\\
\textbf{then} $\connect_v:=true;$
\item[\CRE: (Nodes which must not be connected)]~\\
\textbf{If} $\neg \IsRoot(v) \wedge \connect_v \wedge \neg \ConnectS(v) \wedge \CParent(v) \wedge \dist_{\parent_v}\neq \infty$ 
\textbf{then} $\connect_v:=false;$
\item[\CRG: (Consequence of a deletion)]~\\
\textbf{If} $\neg \IsRoot(v) \wedge \connect_v \wedge \neg \ConnectS(v) \wedge \CParent(v) \wedge \dist_{\parent_v}= \infty $
\textbf{then} $\connect_v:=false; \dist_v:=\infty;$ $\connectpt_v:=false; $\\
\hspace*{1,1cm} send $\InfoMsg_v$ to all $u \in N(v)$ and wait until ($\not \exists u\in N(v), \parent_u=$\\ \hspace*{1,2cm}$\id_v \wedge \connect_u$)

\end{description}
\end{small}
\paragraph{\underline{Update the Steiner tree:}}

Since we use shortest paths to connect members to the existing Steiner
tree, we must maintain distances from members to 
connection points. A connection point is a connected member or a connected node with more than one connected children, i.e. the root of the branch connecting a member. Every connected node updates its distance if it has a better path. So thanks to connection points and distance computation, we maintain a shortest path between a member and the Steiner tree in order to respect the construction in \cite{ImaseWaxman91}. Rule \CRF\/ is used by a connected node to change its variable $\connectpt$ and to become or not a connection point. This rule is executed only if the connected node has a stabilized connection path (i.e. Predicate $\ConnectS$ is \emph{true}).

\begin{small}
\begin{description}
\item[\CRF: (Connected path stabilization)]~\\
\textbf{If} $\neg \IsRoot(v) \wedge \connect_v \wedge \ConnectS(v) \wedge \CParent(v) \wedge \neg \ConnectPtS(v)$\\ 
\textbf{then} \hspace*{0,1cm}\textbf{If} $\member_v$ \textbf{then} $\connectpt_v:=true;$ \\ \hspace*{0,8cm}\textbf{Else} $\connectpt_v := |\{u:u\in N(v) \wedge \parent_u = \id_v \wedge \connect_u\}|>1;$
\end{description}
\end{small}

\section{Correctness and proof in Static setting}
\label{sec:correction}

\begin{definition}[Legitimate state of DST]
\label{def:legitimate_state}
A configuration of algorithm is legitimate iff each process $v \in V$ satisfies the following conditions:
\begin{enumerate}
\item a Steiner tree $T$ spanning the set of members $S$ is constructed;
\item a shortest path connects each member $v \in S$ to the existing tree.
\end{enumerate}
\end{definition}


\begin{lemma}
\label{lem:correct1}
Eventually the node's parent relation constructs a rooted spanning tree in the network.
\end{lemma}

\begin{proof}
Function $\IsRoot(v)$ is a perfect oracle which returns true if $v$
is the root of the tree and false otherwise. So we assume that there
is a time after which only one root exists in the network. Moreover
Rule \CRA\/ is only used by the root to correct its corrupted variables.

Since there is only one root in the network, to have a spanning tree
we must show that each node has one parent and there is no
cycle. First note that each node $v$ could have at each time only one
parent in its neighborhood (see predicate $\CParent(v)$) designed by
variable $\parent_v$, only root has its parent equal to itself. Each
node maintains its level stored in variable $\level_v$ which is updated
by Rules \CRA\/, \ARA\/ and \ARB\/. The level of each node is equal to the level
of its parent plus one, except for the root which has a level at zero
(see Rule \CRA\/). Suppose there is a cycle in the node's parent
relation. This implies that there is a time after which we have a
sequence of nodes with a growing sequence of levels. But there is at
least one node $x$ with a smaller level than its parent $y$ in the
cycle. That is, for $x$ we have $\level_x \neq \level_y +1$ and for $y$
we have $\parent_x=ID_y \wedge \level_x \neq \level_y +1$. So predicate
$\CParent$ is false for $x$ and $y$, thus $x$ and $y$ can execute Rule \ARA\/
to reset their variables and break the cycle. Therefore, there is a
time after which no cycle exists in the structure described by the
node's parent relation. Since there is only one root in the network
(i.e. $\level_v=0$ and $\parent_v=ID_v$) and there is no cycle, thus the
node's parent relation describe one tree spanning the network.
\end{proof}

\begin{lemma}
\label{lem:correct1bis}
Eventually each non-connected node knows its distance to the current tree.
\end{lemma}

\begin{proof}
A node $v$ is connected iff $\connect_v=true$. There is at least one
connected node because the root is always connected (see Rule \CRA\/),
otherwise there is a time where the root corrects its variables using
Rule \CRA\/. According to Lemma \ref{lem:correct1}, a tree spanning the
network is constructed. Let $x$ be a non-connected node, $d_x$ the
distance of the shortest path from $x$ to any connected node and $y$
the neighbor on this shortest path. Suppose $\dist_x>d_x$, thus it
exists a time after which a neighbor offers a better path and $x$ can
execute Rule \ARA\/ because predicate $\Better(x)$ is true. So $x$
corrects $\dist_x$ as the minimum distance in its neighborhood (see
function $\distNotConnect(x)$). Therefore there is a time after which
$\dist_x=d_x$. Moreover, at each time $x$ executes Rule \ARA\/ the variable
$\parent_x$ is modified respectively to variable $\dist_x$ (see function
$\parentNotConnect(x)$) and thus $\parent_x$ stores the neighbor of $x$
which offers to $x$ the shortest path to any connected
node. Therefore, there is a time after which when we have $\dist_x=d_x$
then $\parent_x=y$.
\end{proof}

\begin{lemma}
\label{lem:correct2}
Eventually each Steiner member is linked to root via a connected path.
\end{lemma}

\begin{proof}
A node $v$ is connected iff $\connect_v=true$. There is at least one
connected node because the root is always connected (see Rule \CRA\/),
otherwise there is a time where the root corrects its variables using
Rule \CRA\/. Moreover, according to lemma \ref{lem:correct1}, there is
only one root and a rooted tree spanning the network is
constructed. Thus it exists a path between each member and the root.

To prove the lemma, we first show that for each node $v$ on the path
connecting a member we have $\need_v=true$.

Each node $v$ (except the root) can change the value of its variable
$\need_v$ or $\connect_v$ to true respectively with Rule \CRB\/ and \CRC\/
only when $v$ has no neighbor with a lower distance than its parent
(i.e. $v$ has no better path so \ARA\/ and \ARB\/ are not
executable). Otherwise $\Better(v)$ returns true and Rules \ARA\/ or
\ARB\/ are uppermost used to correct $\dist_v$ and $\parent_v$. So we
suppose that $\Better(v)$ returns false.

Note that for any member $v$ we must have $\need_v=true$ otherwise $v$
executes Rule \CRB\/ to correct $\need_v$. Since there is a path from each
member $v$ to the root, the parent $u$ of a member will execute Rule
\CRB\/ because according to procedure $\Ask(u)$, $u$ has at
least a child $v$ s.t. $\need_v=true$. Thus $u$ changes the value of
its variable $\need_v$ if necessary. Therefore one can show by
induction using the same scheme that for each node $v$ on the path
between a member and the root we have $\need_v=true$.

Each node $v$ (except the root) with $\connect_v=false$ can correct
its variable $\connect_v$ only when Rule \CRB\/ is not executable
(i.e. $\need_v=true$) because predicate $\ConnectS(v)=false$ and
Rule \CRD\/ can not be executed. Since the root $u$ is always connected
(i.e. $\connect_u=true$), each child $v$ of the root with
$\need_v=true$ and $\connect_v=false$ can execute Rule \CRD\/ to change
the value of its variable $\connect_v$ if necessary because predicate
$\ConnectS(v)$ is satisfied. Thus one can show by induction that
for any node on the path between a member and the root we have
$\connect_v=true$.
\end{proof}

\begin{lemma}
\label{lem:correct3}
Eventually $\ConnectPtS(v)$ is true for every connected node $v$ on the path between each member and the root in the network.
\end{lemma}

\begin{proof}
According to Lemma \ref{lem:correct2}, there is a time after which we have paths of connected nodes between members and the root. Note that in this case predicate $\ConnectS(v)$ is true.

Suppose that $\ConnectPtS(v)$ for a connected node $v$ is false. If $v$ is a member then this implies that $\connectpt_v = false$ (see predicate $\ConnectPtS(v)$), so $v$ can execute Rule \CRF\/ to change the value of $\connectpt_v$ to true and we have $\ConnectS(v)=true$. Otherwise, let $v$ be the parent of a member $u$ on the path of connected nodes connecting $u$ to the root. This implies that $\connectpt_v \neq |\{u:u\in N(v) \wedge \parent_u = \id_v \wedge \connect_u\}|>1$ (see predicate $\ConnectPtS(v)$), so $v$ can execute Rule \CRF\/ to update $\connectpt_v$ and we have $\ConnectS(v)=true$. Thus one can show by induction on the height of the tree that it exists a time where $\ConnectS(v)$ is true for every connected node $v$ on the path between each member and the root.
\end{proof}

\begin{lemma}
\label{lem:correct4}
Eventually each member is connected by a shortest path to the current tree.
\end{lemma}

\begin{proof}
Let $T_{i-1}$ be the tree constructed by the algorithm before the
connection of the member $v_i$. To prove the lemma, we must show that
for any member $v_i$ we have a shortest path from $v_i$ to $T_{i-1}$
when $\ConnectPtS(v_i)=true$ and $\Better(v_i)=false$
(i.e. Rule \ARB\/ can not be executed by a member and so there is no
better path to connect the member).

Initially, according to Rule \CRA\/ the root $v_0$ is always connected and
we have $\ConnectPtS(v_0)=true$ and
$\Better(v)=false$ (because $\dist_v=0$). We show by induction on
the number of members that the property is satisfied for each
member. At iteration $1$, let $v_1$ be a not connected member then
according to Lemma \ref{lem:correct1bis} the path $P_1$ from $v_1$ to
$v_0$ in the spanning tree is a shortest path, so there is a time
s.t. $\ConnectPtS(v_1)=true$ (see Lemma
\ref{lem:correct3}) since $P_1$ is a shortest path between $v_1$ and
$v_0$ (i.e. $T_0$), we have $\Better(v_1)=false$, thus the
property is satisfied for $v_1$. We suppose that the tree $T_i$
satisfies the desired property for every member $v_j, j \leq i$. At
iteration $i+1$, when member $v_{i+1}$ is not connected, according to
Lemma \ref{lem:correct1bis} the path $P_{i+1}$ from $v_{i+1}$ to $T_i$
is a shortest path, so there is a time
s.t. $\ConnectPtS(v_{i+1})=true$ (see Lemma
\ref{lem:correct3}). Since $P_{i+1}$ is a shortest path between
$v_{i+1}$ and $T_i$, we have $\Better(v_{i+1})=false$ and the
property is satisfied for $v_{i+1}$.

Note that a member $v_{i+1}$ can create a connection point $u$
(i.e. $\connectpt_u=true$) on the path $P_j$ connecting a
member $v_j, j \leq i$. In this case, the property is still satisfied
for $v_j$ because the path between $u$ and $v_j$ is part of $P_j$ so
it is a shortest path since a subpath of a shortest path is a shortest
path. Moreover, when we have $\connectpt_u=true$ for $u$ then
all nodes on the path between $u$ and $v_j$ update their distance with
Rule \ARB\/ (see predicate $\Better$).
\end{proof}

\begin{lemma}
\label{lem:correct5}
Eventually a Steiner tree is constructed.
\end{lemma}

\begin{proof}
According respectively to Lemmas \ref{lem:correct1} and
\ref{lem:correct2} a spanning tree is constructed (i.e. $S$ is also
spanned) and there is a path of connected nodes between each member
and the root. To prove the lemma we must show that every leaf of $T$
is a member.

Consider the connected node $v$ (i.e. $\need_v=true$ and
$\connect_v=true$), such that $v$ is a leaf of $T$. Since $v$ is a
leaf, this implies that $v$ has no connected child in $T$, so
predicate $\Ask(v)$ is false.\\ Suppose that $v$ is not a
member. Thus $v$ can execute Rule \CRC\/ and change the value of
$\need_v$ to false. As a consequence predicate $\ConnectS(v)$
is false and $v$ can then execute Rule \CRE\/ which changes the value of
$\connect_v$ to false. Therefore $v$ is not connected and is no more
a leaf of $T$. By using the same scheme we can show by induction on
the height of $T$ that every node on a path of connected nodes which
contains no member nodes can not belong to $T$ after a finite bounded
of time.\\ Now suppose that $v$ is a member, the guard of Rule \CRC\/ is
not satisfied so $\need_v$ remains true. Since $\need_v=true$, predicate
$\ConnectS(v)$ remains true too and $v$ is maintained by the
algorithm as a leaf of $T$.
\end{proof}

\begin{lemma}[Convergence]
\label{lem:convergence}
Starting from an illegitimate configuration eventually the algorithm reaches in a finite time a legitimate
configuration.
\end{lemma}

\begin{proof}
Let $C$ be an illegitimate configuration, i.e. $C \not \in
\mathcal{L}$. According to Lemmas \ref{lem:correct1},
\ref{lem:correct4} and \ref{lem:correct5}, in a finite time a
legitimate state is reached for any process $v \in V$. Therefore in a
finite time a legitimate configuration is reached in the network.
\end{proof}

\begin{lemma}[Correction]
The set of legitimate configurations is closed.
\end{lemma}

\begin{proof}
According to the model, $\InfoMsg$ messages are exchanged periodically
with the neighborhood by all nodes in the network, so $\InfoMsg$
messages maintain up to date copies of neighbor states. Thus starting
in a legitimate configuration the algorithm maintains a legitimate
configuration.
\end{proof}

\section{Correctness and proof in Dynamic setting}


In this section, we consider dynamic networks and we prove that topology changes can be correctly treated by extending our algorithm, given in Figure \ref{fig:algo_dynamic}. Moreover, we show that a passage predicate is satisfied during restabilizing execution of given algorithm.

In the following, we define the topology change events, noted $\varepsilon$, that we must consider:
\begin{itemize}
\item an add (resp. a removal) of a member $v$ ($v$ remains in the network) noted ${\tt add}_v$ (resp. ${\tt del}_v$);
\item an add (resp. a removal) of edge $(u,v)$ in the network noted ${\tt recov}_{uv}$ (resp. ${\tt crash}_{uv}$);
\item an add (resp. a removal) of a neighbor node $u$ of $v$ in the network noted ${\tt recov}_u$ (resp. ${\tt crash}_u$).
\end{itemize}


Algorithm given in Figure \ref{fig:algo_dynamic} completes the self-stabilizing algorithm described in precedent sections and allows to a node $v$ to take into account topology change events.
\begin{small}
\begin{figure}[!ht]
\begin{center}
\fbox{
\begin{minipage}{11.7cm}
\begin{description}
\item[\underline{Do forever:}] send $\InfoMsg_v$ to all $u \in N(v)$
\item[\underline{Upon receipt of $\InfoMsg_u$ from $u$:}]$\newline$
use all the rules to correct the local state of $v$ $\newline$
send $\InfoMsg_v$ to all $u \in N(v)$
\item[\underline{Interrupt Section:}]$\newline$
\textbf{If} $\varepsilon$ is a ${\tt del}_v$ event or ($\varepsilon$ is a ${\tt crash}_{uv}$ or ${\tt crash}_u$ event and $\parent_v=ID_u$) \\ \textbf{then}
\hspace*{0.1cm}$\connect_v:=false;$ $\dist_v:=\infty;$ $\connectpt_v:=false; $\\
\hspace*{0.8cm}send $\InfoMsg_v$ to all $u \in N(v)$ $\newline$
\hspace*{0.8cm}\textbf{wait until} $(\not \exists u \in N(v), \parent_u=ID_v \wedge \connect_u);$
\end{description}
\end{minipage}
}
\end{center}
\caption{Algorithm describing message exchanges and treatment of topology change events.}
\label{fig:algo_dynamic}
\end{figure}
\end{small}
In the sequel we suppose that after every topology change the network remains connected. We prove in the next subsection that algorithm of Figure \ref{fig:algo_dynamic} has a superstabilizing property.

\subsection{Correctness under restricted dynamism}

We provide below definitions of the topology change events class $\Lambda$ and passage predicate for protocol given in Figure \ref{fig:algo_dynamic}.

\begin{definition}[Class $\Lambda$ of topology change events]
 ${\tt del}_v, {\tt crash}_{uv}$ and ${\tt crash}_v$ compose the class
 $\Lambda$ of topology change events.
\end{definition}

\begin{definition}[Passage predicate]
\label{def:passage_predicate}
Parent relations can be modified for nodes in the subtree connected by
the removed member, edge or node, and parent relations are not changed
for any other node in the tree.
\end{definition}

%


\begin{lemma}
\label{lem:dyn_suppression}
Starting from a legitimate configuration, if a
member $x$ leaves the set of members $S$ or node $x$ or edge $(y,x)$
is removed from the network then each connected node $v$ in the
subtree of $x$ is disconnected from the tree and a legitimate
configuration is reached by the system.
\end{lemma}

\begin{proof}
According to the description of the complete algorithm, when a member $x$ leaves the set of members $S$ then $x$ changes first its variables as following: $\connect_x=false$ and $\dist_x=\infty$, then $x$ sends its state to its neighborhood and finally $x$ waits until it has no connected child. In the same way, if a node $x$ (resp. edge $(y,x)$ (assume $\parent_x=ID_y$)) is removed from the network then each child $v$ of $x$ (resp. $x$) changes first its variables as following: $\connect_v=false$ and $\dist_v=\infty$ (resp. $\connect_x=false$ and $\dist_x=\infty$), then $v$ (resp. $x$) sends its state to its neighborhood and finally $v$ (resp. $x$) waits until it has no connected child.

When a connected child $u$ of $v$ (resp. of $x$) receives message $\InfoMsg_v$ from $v$ (resp. $\InfoMsg_x$ from $x$), since predicate $\ConnectS(u)$ is false (because $\connect_{parent_u}=false$) and $\dist_{parent_u}=\infty$ the node $u$ executes Rule \CRG\/ changing the variables of $u$ like $v$'s or $x$'s variables, sends its state to its neighborhood and waits until it has no connected child. According to Lemma \ref{lem:no_deadlock}, no node in the subtree of $x$ executing Rule \CRG\/ perpetually waits it has no connected child. As a consequence, after a finite time every connected node $v$ in the subtree of $x$ is no more connected.

Since each node in the subtree of $x$ is not connected, there is at least one of those nodes $v$ such that predicate $\Better(v)$ is true. Thus $v$ can execute Rule \ARA\/. According to Lemmas \ref{lem:correct1} and \ref{lem:correct1bis}, there is a time after which each node in the subtree of $x$ knows its correct shortest path distance to a connected node. Moreover, by Lemmas \ref{lem:correct2} and \ref{lem:correct4} each not connected member will be connected by a shortest path to a connected node in the existing Steiner tree. Therefore, in a finite number of steps the system reaches a legitimate configuration $C' \in \mathcal{L}$.
\end{proof}

\begin{lemma}
\label{lem:2}
The proposed protocol is superstabilizing for the class $\Lambda$ of
topology change events, and the passage predicate (Definition
\ref{def:passage_predicate}) continues to be satisfied while a
legitimate configuration is reached.
\end{lemma}

\begin{proof}
Consider a configuration $\Delta \vdash \mathcal{L}$. Suppose $\varepsilon$ is a removal of edge $(u,v)$ from the network. If $(u,v)$ is not a tree edge then the distances of $u$ and $v$ are not modified neither $u$ nor $v$ changes its parent, thus no parent relation is modified. Otherwise let $\parent_v=u$, $u$'s distance and $u$'s parent are not modified, it is true for any other node not contained in the subtree of $v$ since the distances are not modified (i.e. predicate $\Better$ is not satisfied). However, $u$ is no more a neighbor of $v$ so according to the handling of an edge removal by the algorithm $v$'s variables are reseted. Then $v$ sends its state to its neighborhood and waits until it has no connected child. According to Lemma \ref{lem:dyn_suppression}, all its children will become not connected and eventually change their parent by executing Rule \ARA\/ because there is a better path (i.e. predicate $\Better$ is satisfied). Therefore, only any node in the subtree connected by the edge $(u,v)$ may change its parent relation.

Suppose $\varepsilon$ is a removal of node $u$ from the network. Any node not contained in the subtree of $u$ do not change its parent relation because the distances are not modified (i.e. predicate $\Better$ is not satisfied). Consider each edge $(u,v)$ between $u$ and its child $v$, we can apply the same argument described above for an edge removal. Therefore, only any node contained in the subtree connected by $u$ may change its parent relation.
\end{proof}

A fault which occurs in the network is detected using a distance with an infinity value. To handle a fault, we introduce Rule \CRG\/ to bootstrap connected nodes in the subtree below a faulty node/edge. We show in Lemma \ref{lem:no_deadlock} that even Rule \CRG\/ is executed when no fault occurs in the network then no node perpetually waits (no deadlock) because of Rule \CRG.

\begin{lemma}
Starting from an arbitrary configuration, Rule \CRG\/ introduces no deadlock in the network.
\label{lem:no_deadlock}
\end{lemma}

\begin{proof}
Consider a configuration which simulates the presence of a fault in the network (but there is not really a fault) and allows the execution of Rule \CRG\/ by a node $v$, i.e. $v$ is a connected node and has a not connected parent $u$ with $\dist_{\parent_v}=\infty$. According to Rule \CRG, $v$ becomes a not connected node and sets its distance to infinity (i.e. $\connect_v=false$ and $\dist_v=\infty$), then it sends its state to its neighbors and waits until it has no connected child. There are two cases: (1) $v$ has no connected child or (2) $v$ has at least one connected child. In case (1), $v$ is a leaf of the connected subtree and does not wait. Otherwise, in case (2) the subtree of connected nodes rooted in $v$ has a finite height so we can show by induction that in a finite time every node in the subtree executes Rule \CRG\/. According to case (1), there is no deadlock for the leaves of the connected subtree. Therefore, we can show by induction on the height of the subtree rooted in $v$ that after a finite time there is no connected node and $v$ wakes up.
\end{proof}

\subsection*{Correctness under fully dynamism  assumptions}

In the precedent subsection guarantees are given on the conservation of the tree structure, only for removal topology events. Here, we consider all the different topology change events presented in Section \ref{sec:correction} (i.e. add/removal of members, nodes or edges). We must maintain a quality of service on the weight of the structure reserved to interconnect all members. Therefore, legitimate configurations take into account a global constraint on the Steiner tree weight. As a consequence, we can not give any guarantees on the tree structure during the stabilization of protocol defined by the presented rules and algorithm of Figure \ref{fig:algo_dynamic} (i.e. no passage predicate is satisfied) if an add of a member, node or edge arises in the network. However to maintain a quality of service on the structure weight, we show here that the protocol is able to restabilize when one of the previous mentioned topology change events arises in the network.

Lemma \ref{lem:dyn_suppression} proves that a legitimate configuration is reached starting from an arbitrary configuration if removal topology change events arises in the network. The following lemma considers add topology change events and shows that a legitimate configuration is reached too.

\begin{lemma}
\label{lem:dyn_ajout}
Starting from a legitimate configuration, after a
member add to $S$ or a node or edge add in the network, eventually the
algorithm leads in a finite number of steps to a legitimate
configuration.
\end{lemma}

\begin{proof}
We must consider three cases: an edge add, a node add and the add of a path in the network.

Consider the add of an edge between two existing nodes $u$ and $v$
with a weight $w(u,v)$. If predicate $\Better$ is false for $u$
and $v$ (i.e. $\dist_u \leq \dist_v + w(u,v)$ and $\dist_v \leq \dist_u +
w(u,v)$) then the system is still in a legitimate configuration $C'
\in \mathcal{L}$. Otherwise $\Better$ is true and Rule \ARA\/
(resp. \ARB\/) can be executed if $u$ or $v$ is not connected
(resp. connected) to correct its distance. In the same way, other tree
nodes $u$ or $v$ correct their distances, thus after a finite number
of steps the system reaches a legitimate configuration $C' \in
\mathcal{L}$.

Consider the add of a node $v$ to an existing node $u$ by an edge
$(u,v)$. $v$ corrects its variables by executing Rule \ARA\/. If $v$ is
not a member, variable $\need_v$ is corrected if necessary with Rule \CRC\/
otherwise according to Lemmas \ref{lem:correct2}, \ref{lem:correct3}
and \ref{lem:correct4} $v$ is connected by a shortest path to the
existing tree, which leads the system to a legitimate configuration
$C' \in \mathcal{L}$.

Consider the add of a path $P$. If $P$ is a path between an existing
node $u$ and a new node $v$ then all nodes of $P$ behave like the case
of a node add $v$ to an existing node $u$. Otherwise $P$ is a path
between two existing nodes $u$ and $v$, all nodes of $P$ behave like
the case of a node add to an existing node and $u$ and $v$ behave like
the case of an edge add if $P$ offers a better path. Thus, in a finite
number of steps the system reaches a legitimate configuration $C' \in
\mathcal{L}$.
\end{proof}

\section*{Complexity and Cost Issues}
\begin{theorem}
Using the notation of Theorem~\ref{theo1}, 
Algorithm \STT\/ performs in $O(D\cdot |S|)$ rounds where $D$
 is the current diameter of the
network. It uses $O(\Delta \log n)$ bits of memory in
the send/receive model\footnote{In the classical message passing model
the memory complexity is $O(\log |S|)$}, where $\Delta$ is the current maximal degree of the
network.
\end{theorem}

\begin{proof}
We consider the worst case in which all the tree must be reconstructed
because of topological or member set modifications. Let
$T_i=(V_{T_i},E_{T_i})$ be a tree constructed at some step $i$ of the
algorithm. Our algorithm can be viewed as a special case of a shortest
path tree construction in which all nodes $v \in V_{T_i}$ are
considered as a single virtual root and all nodes $v \not \in V_{T_i}$
computes the shortest distance from this virtual root. So we can show
by induction that the algorithm connects in at most $O(D)$ rounds the
nearest member to the tree $T_i$. Initially when the root $r$ is
stabilized and connected to $T_0$, $r$ initiates a classic shortest
path computation. So after $3D$ rounds the algorithm connects the
nearest member to the root (we need at most $D$ rounds to compute the
shortest path to the root and at most $2D$ rounds for the nodes on the
path to change their states from not connected to connected). We
assume that following the first $3iD$ rounds $i$ members are connected
to the tree $T_i$. We prove that after $3D$ additional rounds $i+1$
members are connected. In at most $D$ rounds all nodes $v \not \in
T_i$ compute their shortest path to $T_i$, in additional $2D$ rounds
all nodes on the path from the nearest member $v \not \in T_i$ to
$T_i$ change their state to connected. So after $3(i+1)D$ rounds $i+1$
members are connected in tree $T_{i+1}$. Thus as $0 \leq i \leq z$ the
algorithm connects all members in at most $O(zD)$ rounds.

In the following we analyze the memory complexity of our solution.
Each node maintains a constant number of 
local variables of size $O(\log n)$ bits. However, due 
to specificity of our model (the send/receive model) the memory 
complexity including the copies of the local neighborhood is
$O(\delta \log n)$ where $\delta$ is the maximal degree of the network.
\end{proof}

 Since we use the shortest distance metric between nodes in the network, any network can be represented by a complete graph so the following Lemma can be applied.

\begin{lemma}[Imase and Waxman \cite{ImaseWaxman91}]
\label{lem:ImaseWaxman91}
Let $G=(V,E)$ be a complete graph with a cost function $C:E \rightarrow \mathbb{R}^+$ satisfying the triangle inequality, and let $S$ be any nonempty subset of $V$ with $|S|=z$. If $2P$ is the cost of an optimal tour for $S$ and $l:V \rightarrow \mathbb{R}^+$ satisfying the following conditions:
\begin{enumerate}
\item $d(u,v) \geq \min(l(u),l(v))$ for all nodes $u,v \in S$, and
\item $l(v) \leq P$ for all nodes $v \in S$,
\end{enumerate}
then $(\sum_{v \in S}l(v))-\max_{v \in S} l(v) \leq (\lceil \log z \rceil)P$.
\end{lemma}

\begin{theorem}\label{theo1}
Let $G=(V,E,w)$ be a dynamic network, and let $S$ be a set of members. Algorithm \STT\/ is a superstabilizing algorithm that returns a steiner tree $T$ for $S$ satisfying
$\frac{W(T)}{W(T^*)} \leq \lceil \log |S| \rceil$,  where $T^*$ is an optimal Steiner tree for $S$.
\end{theorem}

\begin{proof}
Let a set $S$ of members, and $z=|S|$. According to Lemmas \ref{lem:correct4} and \ref{lem:correct5}, when our algorithm completes each member $v \in S$ is connected in $T$ by a shortest path to a node $u$, such that $u$ has been connected in $T$ before $v$. Let $T_{i-1}$ the tree constructed by our algorithm before the connection of a member $v_i \in S$. As in \cite{ImaseWaxman91} (proof of theorem 2), if we let $l(v_i)= \min_{0 \leq j < i} d(v_i,v_j)$ for $1 \leq i \leq z$, then the cost of the path selected by the algorithm to connect $v_i$ to $T_{i-1}$ is less than or equal to $l(v_i)$. Let $l(v_0)=\max_{1 \leq j \leq z} d(v_0,v_j)$, so $l(v_0) \geq \max_{0 \leq j \leq i} l(v_j)$. Thus we have $W(T) \leq (\sum_{j=0}^{z} l(v_j))-l(v_0)$. Moreover for any pair of nodes $v_j,v_k$, according to definition of function $l$ we have $l(v_k) \leq d(v_j,v_k)$ so (1) of lemma \ref{lem:ImaseWaxman91} holds. Note that a tour of set $S$ can be constructed from a Steiner tree for $S$ of cost of $P$ such that the cost of the tour is no more than twice the cost of the Steiner tree. Since $l(v_j) \leq P$ for all $j$, $0 \leq j \leq z$, (2) of lemma \ref{lem:ImaseWaxman91} holds and according to lemma \ref{lem:ImaseWaxman91} the theorem follows.

Since $S$ is a dynamic set of member, we must consider two cases: the add of a member and the removal of a member. Consider the add of a new member $v$ to $S$. By Lemma \ref{lem:dyn_ajout}, the system reaches a legitimate configuration. Thus, $v$ is connected by a shortest path to the existing Steiner tree and $W(T) \leq (\sum_{j=0}^{z} l(v_j))-l(v_0)$ is still satisfied. The same argument is true for the add of a node or an edge of the network. Consider the removal of a member $v$ from $S$. By Lemma \ref{lem:dyn_suppression}, the system reaches a legitimate configuration. Thus, each member $v$ of $S$ is connected by a shortest path to a connected member in the Steiner tree and $W(T) \leq (\sum_{j=0}^{z} l(v_j))-l(v_0)$ is satisfied again. The same argument is true for the removal of a node or an edge of the network. Therefore, considering a dynamic network $G$ and a dynamic set of members the theorem is always satisfied.
\end{proof}

\section{Conclusion}

We propose a self-stabilizing algorithm for the Steiner tree problem, based on the heuristic proposed in \cite{ImaseWaxman91}, and achieves starting from any configuration a competitiveness of $log(z)$ in $O(zD)$ rounds with $z$ the number of members and $D$ the diameter of the network. Additionally, we show that our algorithm works for dynamic networks in which a fault may occur on a node or edge. Moreover, we prove that if a fault occurs in a legitimate configuration our algorithm is superstabilizing and is able to satisfy a "passage predicate" about the tree structure.

For future works, it will be interesting to design a self-stabilizing algorithm in dynamic networks for the Steiner tree problem, which achieves a constant competitiveness of 2. For example, by using the self-stabilizing algorithm proposed in \cite{Kakugawa_Steiner_journal} and extending it for dynamic networks or by using another heuristic.

\bibliographystyle{alpha}
\bibliography{SS_Steiner_Bib}

\end{document}